\documentclass[twoside,11pt]{article}
\usepackage{jair, theapa, rawfonts}

\jairheading{64}{2019}{429-443}{10/18}{2/19}
\ShortHeadings{Revisiting CFR$^+$ and Alternating Updates}
{Burch, Moravcik, \& Schmid}
\firstpageno{429}

\usepackage{amsmath}
\usepackage{amssymb}
\usepackage{mathtools}

\newtheorem{theorem}{Theorem}
\newtheorem{lemma}[theorem]{Lemma}
\newtheorem{corollary}[theorem]{Corollary}
\newenvironment{proof}[1][Proof]{\textbf{#1.} }{\ \rule{0.5em}{0.5em}}
\newtheorem{observation}{Observation}

\DeclareMathOperator*{\expl}{expl}
\newcommand{\defeq}{\vcentcolon=}
\def\etal{\textit{et al.}}

\begin{document}

\title{Revisiting CFR$^+$ and Alternating Updates}

\author{\name Neil Burch \email burchn@google.com \\
  \name Matej Moravcik \email moravcik@google.com \\
  \name Martin Schmid \email mschmid@google.com \\
  \addr DeepMind, 5 New Street Square, \\
  London, EC4A 3TW}

\maketitle

\begin{abstract}
  The CFR$^+$ algorithm for solving imperfect information games is a
  variant of the popular CFR algorithm, with faster empirical
  performance on a range of problems. It was introduced with a
  theoretical upper bound on solution error, but subsequent work
  showed an error in one step of the proof. We provide updated proofs
  to recover the original bound.
\end{abstract}

\section{Introduction}
\label{sec:introduction}

CFR$^+$ was introduced~\cite{Tammelin14} as an algorithm for
approximately solving imperfect information games, and was
subsequently used to essentially solve the game of heads-up limit
Texas Hold'em poker~\cite{SolvingHulhe}. Another paper associated with
the poker result gives a correctness proof for CFR$^+$, showing that
approximation error approaches zero~\cite{15ijcai-cfrplus}.

CFR$^+$ is a variant of the CFR algorithm~\cite{07nips-cfr}, with much
better empirical performance than CFR. One of the CFR$^+$ changes is
switching from simultaneous updates to alternately updating a single
player at a time. A crucial step in proving the correctness of both
CFR and CFR$^+$ is linking regret, a hindsight measurement of
performance, to exploitability, a measurement of the solution quality.

Later work pointed out a problem with the CFR$^+$
proof~\cite{Farina18OnlineConvexOptimization}, noting that the CFR$^+$
proof makes reference to a folk theorem making the necessary link
between regret and exploitability, but fails to satisfy the theorem's
requirements due to the use of alternating updates in
CFR$^+$. Farina~\etal{} give an example of a sequence of updates which
lead to zero regret for both players, but high exploitability.

We state a version of the folk theorem that links alternating update
regret and exploitability, with an additional term in the
exploitability bound relating to strategy improvement. By proving that
CFR and CFR$^+$ generate improved strategies, we can give a new
correctness proof for CFR$^+$, recovering the original bound on
approximation error.

\section{Definitions}

We need a fairly large collection of definitions to get to the
correctness proof. CFR and CFR$^+$ make use of the regret-matching
algorithm~\cite{hart2000simple} and regret-matching$^+$
algorithm~\cite{Tammelin14}, respectively, and we need to show some
properties of these component algorithms. Both CFR and CFR$^+$ operate
on extensive form games, a compact tree-based formalism for describing
an imperfect information sequential decision making problem.

\subsection{Regret-Matching and Regret-Matching$^+$}

Regret-matching is an algorithm for solving the online regret
minimisation problem.  External regret is a hindsight measurement of
how well a policy did, compared to always selecting some action. Given
a set of possible actions $A$, a sequence of value functions
$\boldsymbol{v}^t \in \mathbb{R}^{|A|}$, and sequence of policies
$\boldsymbol{\sigma}^t \in \Delta^{|A|}$, the regret for an action is
\begin{align}
  \label{eq:action-regret}
  \boldsymbol{r}^{t+1} & \defeq \boldsymbol{r}^t + \boldsymbol{v}^t - \boldsymbol{\sigma}^t\cdot\boldsymbol{v}^t & \nonumber \\
  \boldsymbol{r}^0 & \defeq \boldsymbol{0} &
\end{align}
An online regret minimisation algorithm specifies a policy
$\boldsymbol{\sigma}^t$ based on past value functions and policies,
such that $\max_a r^t_a/t\rightarrow{}0$ as $t\rightarrow\infty$.

Let $x^+\defeq\max{(x,0)}$,
$\boldsymbol{x}^+\defeq\left[x_1^+,...,x_n^+\right]$, and
\begin{align}
  \label{eq:rm-normalise}
  \boldsymbol{\sigma}_{\mbox{rm}}(\boldsymbol{x}) & \defeq \left\{ \begin{array}{ll}
    \boldsymbol{x}^+ / (\boldsymbol{1}\cdot\boldsymbol{x}^+) & \mbox{if } \exists a \mbox{ s.t. } x_a > 0 \\
    \boldsymbol{1} / |A| & \mbox{otherwise}
  \end{array} \right.
\end{align}
Then for any $t\ge{}0$, regret-matching uses a policy
\begin{align}
  \label{eq:regret-matching}
  \boldsymbol{\sigma}^t & \defeq \boldsymbol{\sigma}_{\mbox{rm}}(\boldsymbol{r}^t)
\end{align}

Regret-matching$^+$ is a variant of regret-matching that stores a set
of non-negative regret-like values
\begin{align}
  \boldsymbol{q}^{t+1} & \defeq (\boldsymbol{q}^t + \boldsymbol{v}^t - \boldsymbol{\sigma}^t\cdot\boldsymbol{v}^t)^+ & \nonumber \\
  \boldsymbol{q}^0 & \defeq \boldsymbol{0} & \label{eq:action-regret-plus}
\end{align}
and uses the same regret-matching mapping from stored values to policy
\begin{align}
  \label{eq:regret-matching-plus}
  \boldsymbol{\sigma}^t & \defeq \boldsymbol{\sigma}_{\mbox{rm}}(\boldsymbol{q}^t)
\end{align}

\subsection{Extensive Form Games}
An extensive form game~\cite{von1947theory} is a sequential
decision-making problem where players have imperfect (asymmetric)
information. The formal description of an extensive form game is given
by a tuple
$\langle{}H,P,p,\boldsymbol{\sigma}_c,u,\mathcal{I}\rangle{}$.

$H$ is the set of all states $h$, which are a history of actions from
the beginning of the game $\emptyset$.  Given a history $h$ and an
action $a$, $ha$ is the new state reached by taking action $a$ at
$h$. To denote a descendant relationship, we say $h \sqsubseteq j$ if
$j$ can be reached by some (possibly empty) sequence of actions from
$h$, and $h\sqsubset{}j\iff{}h\sqsubseteq{}j,h\ne{}j$.

We will use
$Z\defeq\{h\in{}H\mid\nexists{}j\in{}H\mbox{~s.t.~}h\sqsubset{}j\}$ to
denote the set of terminal histories, where the game is over. We will
use $Z(h)\defeq\{z\in{}Z\mid{}h\sqsubseteq z\}$ to refer to the set of
terminal histories that can be reached from some state $h$.

$A(h)$ gives the set of valid actions at $h\in{}H\setminus{}Z$. We
assume some fixed ordering $a_1,a_2,...,a_{|A|}$ of the actions, so we
can speak about a vector of values or probabilities across actions. $a\prec{}b$ denotes that action $a$ precedes $b$, with $a\prec{}b\iff{}a_i=a,a_j=b,i<j$.

$P$ is the set of players, and
$p:H\setminus{}Z\to{}P\bigcup{}\{c\}$ gives the player that
is acting at state $h$, or the special chance player $c$ for states
where a chance event occurs according to probabilities specified by
$\boldsymbol{\sigma}_c(h)\in\Delta^{|A(h)|}$. Our work is restricted
to two player games, so will say $P=\{1,2\}$.

The utility of a terminal history $z$ for Player $p$ is given by
$u_p(z)$. We will restrict ourselves to zero-sum games, where
$\sum_{p\in{}P}u_p(z)=0$.

A player's imperfect information about the game state is represented
by a partition $\mathcal{I}$ of states $H$ based on player
knowledge. For all information sets $I\in\mathcal{I}$ and all states
$h,j\in{}I$ are indistinguishable to Player $p(h)=p(j)$, with the same
legal actions $A(h)=A(j)$. Given this equality, we can reasonably talk
about $p(I)\defeq{}p(h)$ and $A(I)\defeq{}A(h)$ for any $h\in{}I$. For
any $h$, we will use $I(h)\defeq{}I\in\mathcal{I}$ such that $h\in{}I$
to refer to the information set containing $h$. It is convenient to
group information sets by the acting player, so we will use
$\mathcal{I}_p\defeq\{I\in\mathcal{I}\mid{}p(I)=p\}$ to refer to
Player $p$'s information sets.

We will also restrict ourselves to extensive form games where players
have perfect recall. Informally, Player $p$ has perfect recall if
they do not forget anything they once knew: for all states $h,j$ in
some information set, both $h$ and $j$ passed through the same
sequence of Player $p$ information sets from the beginning of the game
$\emptyset$, and made the same Player $p$ actions.

A strategy $\sigma_p:\mathcal{I}_p\to\Delta^{|A(I)|}$ for Player $p$
gives a probability distribution $\boldsymbol{\sigma}_p(I)$ over legal
actions for Player $p$ information sets. For convenience, let
$\boldsymbol{\sigma}_p(h)\defeq{}\boldsymbol{\sigma}_p(I(h))$. A
strategy profile $\boldsymbol{\sigma}\defeq{}(\sigma_1,\sigma_2)$ is a
tuple of strategies for both players.  Given a profile
$\boldsymbol{\sigma}$, we will use $\sigma_{-p}$ to refer to the
strategy of $p$'s opponent.

Because states are sequences of actions, we frequently need to refer
to various products of strategy action probabilities. Given a strategy
profile $\boldsymbol{\sigma}$,
\begin{align}
  \pi^{\boldsymbol{\sigma}}(h) & \defeq \prod_{ia \sqsubseteq h} \sigma_{p(i)}(h)_a \label{eq:pi-h}
\end{align}
refers to the probability of a game reaching state $h$ when players
sample actions according to $\boldsymbol{\sigma}$ and chance events
occur according to
$\sigma_c$.
\begin{align}
  \pi^{\boldsymbol{\sigma}}(h\mid{}j) & \defeq \prod_{\stackrel{ia \sqsubseteq h}{j \sqsubseteq i}} \sigma_{p(i)}(h)_a \label{eq:pi-h-j}
\end{align}
refers to the probability of a game reaching $h$ given that $j$ was
reached.
\begin{align}
  \pi^{\boldsymbol{\sigma}}_p(h) & \defeq \prod_{\stackrel{ia \sqsubseteq h}{p(h) = p}} \sigma_{p(i)}(h)_a \nonumber \\
  \pi^{\boldsymbol{\sigma}}_{-p}(h) & \defeq \prod_{\stackrel{ia \sqsubseteq h}{p(h) \ne p}} \sigma_{p(i)}(h)_a \label{eq:pi-p-h}
\end{align}
refer to probabilities of Player $p$ or all actors but $p$ making the
actions to reach $h$, given that $p$'s opponent and chance made the
actions in $h$. Note that there is a slight difference in the meaning
of the label $_{-p}$ here, with $\pi^{\boldsymbol{\sigma}}_{-p}$
considering actions by both Player $p$'s opponent and chance, whereas
$\sigma_{-p}$ refers to the strategy of $p$'s opponent.
\begin{align}
  \pi^{\boldsymbol{\sigma}}_p(h\mid{}j) & \defeq \prod_{\stackrel{\stackrel{ia \sqsubseteq h}{j \sqsubseteq i}}{p(h) = p}} \sigma_{p(i)}(h)_a \label{eq:pi-p-h-j}
\end{align}
refers to the probability of Player $p$ making the actions to reach
$h$, given $j$ was reached and $p$'s opponent and chance make the
actions to reach $h$. There are a few useful
relationships:
\begin{align}
  & \pi^{\boldsymbol{\sigma}}(h) = \pi^{\boldsymbol{\sigma}}_p(h)\pi^{\boldsymbol{\sigma}}_{-p}(h) & \nonumber \\
  & \forall j \sqsubseteq h,~ \pi^{\boldsymbol{\sigma}}(h) = \pi^{\boldsymbol{\sigma}}(j)\pi^{\boldsymbol{\sigma}}(h\mid{}j) &
\end{align}

The expected utility of a strategy profile $\boldsymbol{\sigma}$ is
\begin{align}
  & u^{\boldsymbol{\sigma}}_p \defeq \sum_{z \in Z} \pi^{\boldsymbol{\sigma}}(z) u_p(z) &
\end{align}
The counterfactual value of a history or information set are defined as
\begin{align}
  \label{eq:cfv}
  & v^{\boldsymbol{\sigma}}_p(h) \defeq \sum_{z \in Z(h)} \pi^{\boldsymbol{\sigma}}_{-p}(z)\pi^{\boldsymbol{\sigma}}_p(z\mid{}h) u_p(z) & \nonumber \\
  & \boldsymbol{v}^{\boldsymbol{\sigma}}(I) \defeq \sum_{h \in I} (v^{\boldsymbol{\sigma}}_{p(h)}(ha_1), ...,  v^{\boldsymbol{\sigma}}_{p(h)}(ha_{|A(I)|})) &
\end{align}
For later convenience, we will assume that for each player there
exists an information set $I^\emptyset_p$ at the beginning of the
game, containing a single state with a single action, leading to the
rest of the game. This lets us say that
$u^{\boldsymbol{\sigma}}_p=v^{\boldsymbol{\sigma}}(I^\emptyset_p)_{a_0}$.

Given a sequence $\sigma^0_p, ..., \sigma^t_p$ of strategies, we
denote the average strategy from $a$ to $b$ as
\begin{align}
  \bar{\sigma}^{[a,b]}_p \defeq \sum_{i=a}^b \frac{\sigma^i_p}{b-a+1}
\end{align}
Given a sequence $\boldsymbol{\sigma}^0, ..., \boldsymbol{\sigma}^{t-1}$
of strategy profiles, we denote the average Player $p$ regret as
\begin{align}
  r^t_p & \defeq \max_{\sigma^*_p}\sum_{i=0}^{t-1} (u^{(\sigma^*_p,\sigma^i_{-p})}_p - u^{\boldsymbol{\sigma}^i}_p) / t \nonumber \\
  & = \max_{\sigma^*_p} u^{(\sigma^*_p,\bar{\sigma}^{[0,t-1]}_{-p})}_p - \sum_{i=0}^{t-1} u^{\boldsymbol{\sigma}^i}_p / t \label{eq:avg-regret-strategy}
\end{align}
The exploitability of a strategy profile $\boldsymbol{\sigma}$ is a
measurement of how much expected utility each player could gain by
switching their strategy:
\begin{align}
  \expl(\boldsymbol{\sigma}) & \defeq \max_{\sigma^*_1}u^{(\sigma^*_1,\sigma_2)}_1 - u^{\boldsymbol{\sigma}}_1 + \max_{\sigma^*_2}u^{(\sigma_1,\sigma^*_2)}_2 - u^{\boldsymbol{\sigma}}_2 \nonumber \\
  & = \max_{\sigma^*_1}u^{(\sigma^*_1,\sigma_2)}_1 + \max_{\sigma^*_2}u^{(\sigma_1,\sigma^*_2)}_2 & \mbox{by zero-sum} \label{eq:exploitability} 
\end{align}
Achieving zero exploitability -- a Nash equilibrium~\cite{Nash1950} --
is possible. In two player, zero-sum games, finding a strategy with
low exploitability is a reasonable goal for good play.

\subsection{CFR and CFR$^+$}
CFR and its variant CFR$^+$ are both
algorithms for finding an extensive form game strategy with low
exploitability. They are all iterative self-play algorithms that track
the average of a current strategy that is based on many loosely
coupled regret minimisation problems.

CFR and CFR$^+$ track regret-matching values $\boldsymbol{r}^t(I)$ or
regret-matching$^+$ values $\boldsymbol{q}^t(I)$ respectively, for all
$I\in\mathcal{I}$. At time $t$, CFR and CFR$^+$ use strategy profile
$\boldsymbol{\sigma}^t(I)\defeq\boldsymbol{\sigma}_{\mbox{rm}}(\boldsymbol{r}^t(I))$
and
$\boldsymbol{\sigma}^t(I)\defeq\boldsymbol{\sigma}_{\mbox{rm}}(\boldsymbol{q}^t(I))$,
respectively.  When doing alternating updates, with the first update
done by Player 1, the values used for updating regrets are
\begin{align}
  \label{eq:cfr-value}
  & \boldsymbol{v}^t(I) \defeq \left\{ \begin{array}{ll}
    \boldsymbol{v}^{\boldsymbol{\sigma}^t}(I) & \mbox{if $p(I)=1$} \\
    \boldsymbol{v}^{(\sigma^{t+1}_1,\sigma^t_2)}(I) & \mbox{if $p(I)=2$}
  \end{array} \right.
\end{align}
and the output of CFR is the profile of average strategies
$(\bar{\sigma}^{[1,t]}_1,\bar{\sigma}^{[0,t-1]}_2)$, while the output
of CFR$^+$ is the profile of weighted average strategies
$(\frac{2}{t^2+t}\sum_{i=1}^{t}i\sigma^i_1,\frac{2}{t^2+t}\sum_{i=0}^{t-1}(i+1)\sigma^i_2)$.

\section{Theoretical Results}
\label{sec:proof}

The CFR$^+$ proof of correctness~\cite{15ijcai-cfrplus} references a
folk theorem that links regret and exploitability. Farina~\etal{} show
that the folk theorem only applies to simultaneous updates, not
alternating updates, giving an example of a sequence of alternating
updates with no regret but constant
exploitability~\cite{Farina18OnlineConvexOptimization}. Their
observation is reproduced below using the definitions from this work.

\begin{observation}
  \label{obs:alternating-counterexample}
  Let $P=\{X,Y\}$, $A=\{0,1\}$, and $Z=\{00, 01, 10, 11\}$. A game
  consists of each player selecting one action. Let $u_X(11)=1$, and
  $u_X(z)=0$ for all $z\ne{}11$. Consider the sequence of strategies
  $\sigma^t_X=\sigma^t_Y=t\bmod{}2$, with Player $X$ regrets computed
  using $\boldsymbol{v}^{(\sigma_X^t,\sigma_Y^t)}$ and Player $Y$
  regrets computed using
  $\boldsymbol{v}^{(\sigma_X^{t+1},\sigma_Y^t)}$. Then at any time
  $2T$ the accumulated regret for both players is 0 and the average
  strategy is
  $\bar{\sigma}^{[1,2T]}_X=\bar{\sigma}^{[0,2T-1]}_Y=\boldsymbol{0.5}$,
  with exploitability
  $\expl(\bar{\sigma}^{[1,2T]}_X,\bar{\sigma}^{[0,2T-1]}_Y)=0.5$. So
  both players have 0 regret, but the exploitability does not approach
  0.
\end{observation}

As a first step in correcting the CFR$^+$ proof, we introduce an
analogue of the folk theorem, linking alternating update regret and
exploitability.

\begin{theorem}
  \label{thm:folk-theorem}
  Let $\boldsymbol{\sigma}^t$ be the strategy profile at some time
  $t$, and $r^t_p$ be the regrets computed using alternating updates
  so that Player 1 regrets are updated using
  $\boldsymbol{v}^{(\sigma^t_1,\sigma^t_2)}$ and Player 2 regrets are
  updated using $\boldsymbol{v}^{(\sigma^{t+1}_1,\sigma^t_2)}$. If the
  regrets are bounded by $r^t_p \le \epsilon_p$, then the
  exploitability of
  $(\bar{\sigma}^{[1,t]}_1,\bar{\sigma}^{[0,t-1]}_2)$ is bounded by
  $\epsilon_1+\epsilon_2-\frac{1}{t}\sum_{i=0}^{t-1}(u^{(\sigma^{i+1}_1,\sigma^i_2)}_1-u^{(\sigma^i_1,\sigma^i_2)}_1)$.
\end{theorem}
\begin{proof}
  Consider the sum of regrets for both players, $r^t_1+r^t_2$
  \begin{align}
    = & \max_{\sigma^*_1} u^{(\sigma^*_1,\bar{\sigma}^{[0,t-1]}_2)}_1 - \frac{1}{t} \sum_{i=0}^{t-1} u^{(\sigma^i_1,\sigma^i_2)}_1 + \max_{\sigma^*_2} u^{(\bar{\sigma}^{[1,t]}_1,\sigma^*_2)}_2 - \frac{1}{t} \sum_{i=0}^{t-1} u^{(\sigma^{i+1}_1,\sigma^i_2)}_2 & \mbox{by Eq.~\ref{eq:avg-regret-strategy}} \nonumber \\
    = & \expl{(\bar{\sigma}^{[1,t]}_1,\bar{\sigma}^{[0,t-1]}_2)} - \frac{1}{t} \sum_{i=0}^{t-1} ( u^{(\sigma^i_1,\sigma^i_2)}_1 + u^{(\sigma^{i+1}_1,\sigma^i_2)}_2 ) & \mbox{by Eq.~\ref{eq:exploitability}} \nonumber
  \intertext{Given $r^t_p\le\epsilon_p$ for all players $p$, we have $\expl{(\bar{\sigma}^{[1,t]}_1,\bar{\sigma}^{[0,t-1]}_2)}$}
    \le & \epsilon_1 + \epsilon_2 + \frac{1}{t} \sum_{i=0}^{t-1} ( u^{(\sigma^i_1,\sigma^i_2)}_1 + u^{(\sigma^{i+1}_1,\sigma^i_2)}_2 ) \nonumber \\
    = & \epsilon_1 + \epsilon_2 + \frac{1}{t} \sum_{i=0}^{t-1} ( u^{(\sigma^i_1,\sigma^i_2)}_1 - u^{(\sigma^{i+1}_1,\sigma^i_2)}_1 ) & \mbox{by zero-sum} \nonumber \\
    = & \epsilon_1 + \epsilon_2 - \frac{1}{t} \sum_{i=0}^{t-1} ( u^{(\sigma^{i+1}_1,\sigma^i_2)}_1 - u^{(\sigma^i_1,\sigma^i_2)}_1 ) & \nonumber
  \end{align}
\end{proof}

The gap between regret and exploitability in
Observation~\ref{obs:alternating-counterexample} is now apparent as a
trailing sum in Theorem~\ref{thm:folk-theorem}. Each term in the sum
measures the improvement in expected utility for Player 1 from time
$t$ to time $t+1$. Motivated by this sum, we show that
regret-matching, CFR, and their $^+$ variants generate new policies
which are not worse than the current policy. Using these constraints,
we construct an updated correctness proof for CFR$^+$.

\subsection{Regret-Matching and Regret-Matching$^+$ Properties}

We will show that when using regret-matching or regret-matching$^+$,
the expected utility $\boldsymbol{\sigma}^{t+1}\cdot\boldsymbol{v}^t$
is never less than $\boldsymbol{\sigma}^t\cdot\boldsymbol{v}^t$. To do
this, we will need to show these algorithms have a couple of other
properties. We start by showing that once there is at least one
positive stored regret or regret-like value, there will always be a
positive stored value.

\begin{lemma}
  \label{thm:positive-regret}
  For any $t$, let $\boldsymbol{s}^t$ be the stored value
  $\boldsymbol{r}^t$ used by regret-matching or $\boldsymbol{q}^t$
  used by regret-matching$^+$, and $\boldsymbol{\sigma}^t$ be the
  associated policy.  Then for all $t$ where $\exists a \in A$ such
  that $s^t_a > 0$, there $\exists b \in A$ such that $s^{t+1}_{b}>0$.
\end{lemma}
\begin{proof}
  Consider any time $t$ where $\exists{}a\in{}A$ such that
  $s^t_a>0$. The policy at time $t$ is then
  \begin{align}
    & \boldsymbol{\sigma}^t=\boldsymbol{s}^{t,+}/(\boldsymbol{1}\cdot\boldsymbol{s}^{t,+}) & \mbox{by Eqs.~\ref{eq:rm-normalise},~\ref{eq:regret-matching},~\ref{eq:regret-matching-plus}} \label{eq:positive-regret-strategy}
  \end{align}
  Consider the stored value $s^{t+1}_a$. With regret-matching
  $s^{t+1}_a=r^{t+1}_a=r^t_a+v^t_a-\boldsymbol{\sigma}^t\cdot\boldsymbol{v}^t$
  by Equation~\ref{eq:action-regret}, and with regret-matching$^+$
  $s^{t+1}_a=q^{t+1}_a=(q^t_a+v^t_a-\boldsymbol{\sigma}^t\cdot\boldsymbol{v}^t)^+$
  by Equation~\ref{eq:action-regret-plus}.  For both algorithms, the
  value of $s^{t+1}_a$ depends on
  $v^t_a-\boldsymbol{\sigma}^t\cdot\boldsymbol{v}^t$. There are two
  cases:
  \begin{enumerate}
  \item $v^t_a - \boldsymbol{\sigma}^t\cdot\boldsymbol{v}^t \ge 0$
    \begin{align}
      & s^{t+1}_a > 0 & \mbox{by Lemma assumption, Eq.~\ref{eq:action-regret},~\ref{eq:action-regret-plus}} \nonumber
    \end{align}
    
  \item $v^t_a - \boldsymbol{\sigma}^t\cdot\boldsymbol{v}^t < 0$
    \begin{align}
      & \sigma^t_a(v^t_a - \boldsymbol{\sigma}^t \cdot \boldsymbol{v}^t) < 0 & \mbox{by $s^t_a>0$, Eq.~\ref{eq:positive-regret-strategy}} \nonumber \\
      & 0 < \boldsymbol{\sigma}^t \cdot \boldsymbol{v}^t - \boldsymbol{\sigma}^t \cdot \boldsymbol{v}^t - \sigma^t_a(v^t_a - \boldsymbol{\sigma}^t \cdot \boldsymbol{v}^t) & \nonumber \\
      & 0 < \sum_{b \in A} \bigl( \sigma^t_b (v^t_b - \boldsymbol{\sigma}^t \cdot \boldsymbol{v}^t)\bigr) - \sigma^t_a(v^t_a - \boldsymbol{\sigma}^t \cdot \boldsymbol{v}^t) & \mbox{by $\sum_{b \in A}\sigma^t_b = 1$} \nonumber \\
      & 0 < \sum_{b \ne a} \bigl( \sigma^t_b (v^t_b - \boldsymbol{\sigma}^t \cdot \boldsymbol{v}^t)\bigr) & \nonumber \\
      & \exists b \mbox{ s.t. } \sigma^t_b > 0,~ v^t_b - \boldsymbol{\sigma}^t\cdot\boldsymbol{v}^t > 0 & \mbox{by $\sigma^t_{a'} \ge 0 $ for all $a' \in A$} \nonumber \\
      & s^t_b > 0,~ v^t_b - \boldsymbol{\sigma}^t\cdot\boldsymbol{v}^t > 0 & \mbox{by Eq.~\ref{eq:positive-regret-strategy}} \nonumber \\
      & s^{t+1}_b > 0 & \mbox{by Eqs.~\ref{eq:action-regret},~\ref{eq:action-regret-plus}} \nonumber
    \end{align}
  \end{enumerate}
  In both cases, $\exists b$ such that $s^{t+1}_b > 0$.
\end{proof}

There is a corollary to Lemma~\ref{thm:positive-regret}, that
regret-matching and regret-matching$^+$ never switch back to playing
the default uniform random policy once they switch away from it.

\begin{corollary}
  When using regret-matching or regret-matching$^+$, if there exists a
  time $t$ such that
  $\sigma^t=\boldsymbol{s}^{t,+}/(\boldsymbol{1}\cdot\boldsymbol{s}^{t,+})$
  where $\boldsymbol{s}^t$ are the stored regrets $\boldsymbol{r}^t$
  or regret-like values $\boldsymbol{q}^t$ at time $t$, then
  $\sigma^{t'}=\boldsymbol{s}^{t',+}/(\boldsymbol{1}\cdot\boldsymbol{s}^{t',+})$
  for all $t'\ge{}t$.
\end{corollary}
\begin{proof}
  Assume that at some time $t$,
  $\sigma^t=\boldsymbol{s}^{t,+}/(\boldsymbol{1}\cdot\boldsymbol{s}^{t,+})$.
  We can show by induction that
  $\sigma^{t'}=\boldsymbol{s}^{t',+}/(\boldsymbol{1}\cdot\boldsymbol{s}^{t',+})$
  for all $t'\ge{}t$.  The base case $t'=t$ of the hypothesis holds by
  assumption. Now, assume that
  $\sigma^{t'}=\boldsymbol{s}^{t',+}/(\boldsymbol{1}\cdot\boldsymbol{s}^{t',+})$
  for some time $t'\ge{}t$. We have
  \begin{align}
    & \exists a \in A \mbox{ s.t. } s^{t'}_a > 0 & \mbox{by Eq.~\ref{eq:rm-normalise}} \nonumber \\
    & \exists b \in A \mbox{ s.t. } s^{t'+1}_b > 0 & \mbox{by Lemma~\ref{thm:positive-regret}} \nonumber \\
    & \sigma^{t'+1} = \boldsymbol{s}^{t'+1,+}/(\boldsymbol{1}\cdot\boldsymbol{s}^{t'+1,+}) & \mbox{by Eq.~\ref{eq:rm-normalise}} \nonumber
  \end{align}
  Therefore, by induction the hypothesis holds for all $t'\ge{}t$.
\end{proof}

\begin{lemma}
  \label{thm:rm-positive-delta}
  For any $t$, let $\boldsymbol{s}^t$ be the stored value
  $\boldsymbol{r}^t$ used by regret-matching or $\boldsymbol{q}^t$
  used by regret-matching$^+$, and $\boldsymbol{\sigma}^t$ be the
  associated policy. Then for all $t$ and $a \in A$, \\
  $(s^{t+1,+}_a-s^{t,+}_a)(v^t_a-\boldsymbol{\sigma}^t\cdot\boldsymbol{v}^t)\ge{}0$.
\end{lemma}
\begin{proof}
  Consider whether $v^t_a - \boldsymbol{\sigma}^t \cdot \boldsymbol{v}^t$ is positive. There are two cases.
  \begin{enumerate}
  \item $v^t_a - \boldsymbol{\sigma}^t \cdot \boldsymbol{v}^t \le 0$
    
    For regret-matching, where $s^t_a=r^t_a$, we have
    \begin{align}
      & r^{t+1}_a = r^t_a + v^t_a - \boldsymbol{\sigma}^t \cdot \boldsymbol{v}^t & \mbox{by Eq.~\ref{eq:action-regret}} \nonumber \\
      & r^{t+1}_a \le r^t_a & \nonumber \\
      & r^{t+1,+}_a \le r^{t,+}_a & \nonumber
    \end{align}
    For regret-matching$^+$, where $s^t_a=q^t_a$, we have
    \begin{align}
      & q^{t+1,+}_a = (q^t_a + v^t_a - \boldsymbol{\sigma}^t \cdot \boldsymbol{v}^t)^+ & \mbox{by Eq.~\ref{eq:action-regret-plus}} \nonumber \\
      & q^{t+1,+}_a = (q^{t,+}_a + v^t_a - \boldsymbol{\sigma}^t \cdot \boldsymbol{v}^t)^+ & \mbox{by Eq.~\ref{eq:action-regret-plus}} \nonumber \\
      & q^{t+1,+}_a \le q^{t,+}_a & \mbox{by monotonicity of $(\cdot)^+$} \nonumber
    \end{align}
    Therefore for both algorithms we have
    \begin{align}
      s^{t+1,+}_a - s^{t,+}_a \le 0 \nonumber \\
      (s^{t+1,+}_a - s^{t,+}_a) (v^t_a - \boldsymbol{\sigma}^t \cdot \boldsymbol{v}^t) \ge 0 \nonumber
    \end{align}

  \item $v^t_a - \boldsymbol{\sigma}^t \cdot \boldsymbol{v}^t > 0$
    \begin{align}
      & s^{t+1}_a = s^t_a + v^t_a - \boldsymbol{\sigma}^t \cdot \boldsymbol{v}^t & \mbox{by Eqs.~\ref{eq:action-regret},~\ref{eq:action-regret-plus}} \nonumber \\
      & s^{t+1}_a > s^t_a \nonumber \\
      & s^{t+1,+}_a \ge s^{t,+}_a \nonumber \\
      & (s^{t+1,+}_a - s^{t,+}_a)(v^t_a - \boldsymbol{\sigma}^t \cdot \boldsymbol{v}^t) \ge 0 \nonumber
    \end{align}
  \end{enumerate}
  In both cases, we have $(s^{t+1,+}_a-s^{t,+}_a)(v^t_a-\boldsymbol{\sigma}^t\cdot\boldsymbol{v}^t)\ge{}0$.
\end{proof}

\begin{theorem}
  \label{thm:rm-improvement}
  If $\boldsymbol{\sigma}^0,\boldsymbol{\sigma}^1,...$ is the sequence
  of regret-matching or regret-matching$^+$ policies generated from
  a sequence of value functions
  $\boldsymbol{v}^0,\boldsymbol{v}^1,...$, then for all $t$,
  $\boldsymbol{\sigma}^{t+1}\cdot\boldsymbol{v}^t\ge\boldsymbol{\sigma}^t\cdot\boldsymbol{v}^t$.
\end{theorem}
\begin{proof}
Let $\boldsymbol{s}^t$ be the stored value $\boldsymbol{r}^t$ used by
regret-matching or $\boldsymbol{q}^t$ used by regret-matching$^+$.
Consider whether all components of $\boldsymbol{s}^t$ or
$\boldsymbol{s}^{t+1}$ are negative. By
Lemma~\ref{thm:positive-regret} we know that it can not be the case
that $\exists{}a~s^t_a>0$ and $\forall{}b~s^{t+1}_b\le{}0$. This
leaves three cases.
  \begin{enumerate}
  \item $\forall a~ s^t_a \le 0$ and $\forall a~ s^{t+1}_a \le 0$
    \begin{align}
      & \boldsymbol{\sigma}^t = \boldsymbol{\sigma}^{t+1} = \boldsymbol{1}/|A| & \mbox{by Eqs.~\ref{eq:rm-normalise},~\ref{eq:regret-matching},~\ref{eq:regret-matching-plus}} \nonumber \\
      & \boldsymbol{\sigma}^{t+1} \cdot \boldsymbol{v}^t = \boldsymbol{\sigma}^t \cdot \boldsymbol{v}^t & \nonumber
    \end{align}

  \item $\forall a~ s^t_a \le 0$ and $\exists a~ s^{t+1}_a > 0$
    \begin{align}
      & \boldsymbol{\sigma}^{t+1} = \boldsymbol{s}^{t+1,+} / (\boldsymbol{1} \cdot \boldsymbol{s}^{t+1,+}) & \mbox{by Eqs.~\ref{eq:rm-normalise},~\ref{eq:regret-matching},~\ref{eq:regret-matching-plus}} \nonumber \\
      & \forall b,~ \sigma^{t+1}_b > 0 \implies s^{t+1}_b > 0 & \nonumber \\
      & \forall b,~ \sigma^{t+1}_b > 0 \implies v^t_b > \boldsymbol{\sigma}^t \cdot \boldsymbol{v}^t & \mbox{by Eqs.~\ref{eq:action-regret},~\ref{eq:action-regret-plus}} \nonumber \\
      & \sum_{b\in{}A} \sigma^{t+1}_b v^t_b > \sum_b \sigma^{t+1}_b \boldsymbol{\sigma}^t \cdot \boldsymbol{v}^t & \mbox{by $\sigma^{t+1}_b \ge 0$} \nonumber \\
      & \boldsymbol{\sigma}^{t+1} \cdot \boldsymbol{v}^t > \boldsymbol{\sigma}^t \cdot \boldsymbol{v}^t & \mbox{by $\sum_{b\in{}A} \sigma^{t+1}_b = 1$} \nonumber
    \end{align}

  \item $\exists a~ s^t_a > 0$ and $\exists b~ s^{t+1}_b > 0$ \\
    Let
    \begin{align}
      \boldsymbol{\sigma}(\boldsymbol{w}) \defeq \boldsymbol{w}^+/(\boldsymbol{w}^+ \cdot \boldsymbol{1}) \label{eq:rm-improvement-sigma}
    \end{align}
    Then we have
    \begin{align}
      & \boldsymbol{\sigma}^t = \boldsymbol{\sigma}(\boldsymbol{s}^t),~ \boldsymbol{\sigma}^{t+1} = \boldsymbol{\sigma}(\boldsymbol{s}^{t+1}) & \mbox{by Eqs.~\ref{eq:rm-normalise},~\ref{eq:regret-matching},~\ref{eq:regret-matching-plus}} \label{eq:rm-improvement-sigma-equality}
    \end{align}
    Consider any ordering $a_1, a_2, ..., a_{|A|}$ of actions such that
    $b\prec{}a$. Let
    \begin{align}
       w^i_j & \defeq \left\{ \begin{array}{ll}
        s^{t+1}_j & j \le i \\
        s^t_j & j > i
      \end{array} \right. \label{eq:rm-improvement-partial-vector}
    \end{align}
    Note that $\forall{}i<a$, $w^i_a=s^t_a>0$, and $\forall{}i\ge{}a$,
    $w^i_b=s^{t+1}_b>0$, so that
    $\boldsymbol{\sigma}(\boldsymbol{w}^i)$ is always well-defined. We
    can show by induction that for all $i\ge0$
    \begin{align}
      & \boldsymbol{\sigma}(\boldsymbol{w}^i)\cdot\boldsymbol{v}^t\ge\boldsymbol{\sigma}^t\cdot\boldsymbol{v}^t & \label{eq:rm-improvement-hypothesis}
      \end{align}
    For the base case of $i=0$, we have
    \begin{align}
      & \boldsymbol{w}^0 = \boldsymbol{s}^t & \mbox{by Eq.~\ref{eq:rm-improvement-partial-vector}} \nonumber \\
      & \boldsymbol{\sigma}(\boldsymbol{w}^0) \cdot \boldsymbol{v}^t = \boldsymbol{\sigma}(\boldsymbol{s}^t) \cdot \boldsymbol{v}^t & \nonumber \\
      & \boldsymbol{\sigma}(\boldsymbol{w}^0) \cdot \boldsymbol{v}^t = \boldsymbol{\sigma}^t \cdot \boldsymbol{v}^t & \mbox{by Eq.~\ref{eq:rm-improvement-sigma-equality}} \nonumber
    \end{align}
    Now assume that Equation~\ref{eq:rm-improvement-hypothesis} holds
    for some $i \ge 0$. By construction,
    \begin{align}
      & \forall{}j\ne{}{i+1},~w^{i+1}_j=w^i_j & \mbox{by Eq.~\ref{eq:rm-improvement-partial-vector}} \label{eq:rm-improvement-partial-vector-similarity}
    \end{align}
    For notational convenience, let
    $\Delta_w\defeq{}w^{i+1,+}_{i+1}-w^{i,+}_{i+1}=s^{t+1,+}_{i+1}-s^{t,+}_{i+1}$.
    \begin{align}
      & \boldsymbol{\sigma}(\boldsymbol{w}^{i+1}) \cdot \boldsymbol{v}^t - \boldsymbol{\sigma}^t \cdot \boldsymbol{v}^t & \nonumber \\
      & = \frac{\boldsymbol{w}^{i+1,+} \cdot \boldsymbol{v}^t}{\boldsymbol{w}^{i+1,+} \cdot \boldsymbol{1}} - \boldsymbol{\sigma}^t \cdot \boldsymbol{v}^t & \mbox{by Eq.~\ref{eq:rm-improvement-sigma}} \nonumber \\
      & = \frac{\Delta_wv^t_{i+1} + \boldsymbol{w}^{i,+} \cdot \boldsymbol{v}^t}{\Delta_w+\boldsymbol{w}^{i,+} \cdot \boldsymbol{1}} - \boldsymbol{\sigma}^t \cdot \boldsymbol{v}^t & \mbox{by Eqs.~\ref{eq:rm-improvement-partial-vector},~\ref{eq:rm-improvement-partial-vector-similarity}} \nonumber \\
      & = \frac{\Delta_wv^t_{i+1} + (\boldsymbol{w}^{i,+} \cdot \boldsymbol{1})\boldsymbol{\sigma}(\boldsymbol{w}^{i})\cdot \boldsymbol{v}^t}{\Delta_w+\boldsymbol{w}^{i,+} \cdot \boldsymbol{1}} - \boldsymbol{\sigma}^t \cdot \boldsymbol{v}^t & \mbox{by Eq.~\ref{eq:rm-improvement-sigma}} \nonumber \\
      & \ge \frac{\Delta_wv^t_{i+1} + (\boldsymbol{w}^{i,+} \cdot \boldsymbol{1})\boldsymbol{\sigma}^t\cdot \boldsymbol{v}^t}{\Delta_w+\boldsymbol{w}^{i,+} \cdot \boldsymbol{1}} - \boldsymbol{\sigma}^t \cdot \boldsymbol{v}^t & \mbox{by ind. hypothesis} \nonumber \\
      & = \frac{\Delta_wv^t_{i+1} + (\boldsymbol{w}^{i,+} \cdot \boldsymbol{1})\boldsymbol{\sigma}^t\cdot \boldsymbol{v}^t}{\Delta_w+\boldsymbol{w}^{i,+} \cdot \boldsymbol{1}} - \frac{\Delta_w\boldsymbol{\sigma}^t \cdot \boldsymbol{v}^t + (\boldsymbol{w}^{i,+} \cdot \boldsymbol{1})\boldsymbol{\sigma}^t \cdot \boldsymbol{v}^t}{\Delta_w+\boldsymbol{w}^{i,+} \cdot \boldsymbol{1}} & \nonumber \\
      & = \frac{\Delta_w(v^t_{i+1} - \boldsymbol{\sigma}^t \cdot \boldsymbol{v}^t)}{\Delta_w+\boldsymbol{w}^{i,+} \cdot \boldsymbol{1}} & \nonumber \\
      & \ge 0 & \mbox{by Lemma~\ref{thm:rm-positive-delta}} \nonumber
    \end{align}
    $\boldsymbol{\sigma}(\boldsymbol{w}^{i+1})\cdot\boldsymbol{v}^t\ge\boldsymbol{\sigma}^t\cdot\boldsymbol{v}^t$,
    so by induction Equation~\ref{eq:rm-improvement-hypothesis} holds
    for all $i\ge{}0$. In particular, we can now say
    \begin{align}
      & \boldsymbol{\sigma}(\boldsymbol{w}^{|A|}) \cdot \boldsymbol{v}^t \ge \boldsymbol{\sigma}^t \cdot \boldsymbol{v}^t & \nonumber \\
      & \boldsymbol{\sigma}(\boldsymbol{s}^{t+1}) \cdot \boldsymbol{v}^t \ge \boldsymbol{\sigma}^t \cdot \boldsymbol{v}^t & \mbox{by Eq.~\ref{eq:rm-improvement-partial-vector}} \nonumber \\
      & \boldsymbol{\sigma}^{t+1} \cdot \boldsymbol{v}^t \ge \boldsymbol{\sigma}^t \cdot \boldsymbol{v}^t & \mbox{by Eq.~\ref{eq:rm-improvement-sigma-equality}} \nonumber
    \end{align}
  \end{enumerate}
  In all cases, we have
  $\boldsymbol{\sigma}^{t+1}\cdot\boldsymbol{v}^t\ge\boldsymbol{\sigma}^t\cdot\boldsymbol{v}^t$.
\end{proof}

\subsection{CFR and CFR$^+$ Properties}

We now show that CFR and CFR$^+$ have properties that are similar to
Theorem~\ref{thm:rm-improvement}. After a player updates their
strategy, that player's counterfactual value does not decrease for any
action at any of their information sets. Similarly, the expected value
of the player's new strategy does not decrease. Finally, using the
property of non-decreasing value, we give an updated proof of an
exploitability bound for CFR$^+$.

\begin{lemma}
  \label{thm:cfr-infoset-improvement}
  Let $p$ be the player that is about to be updated in CFR or CFR$^+$ at
  some time $t$. Let $\sigma^t_p$ be the current strategy for $p$, and
  $\sigma_o$ be the opponent strategy $\sigma^t_{-p}$ or
  $\sigma^{t+1}_{-p}$ used by Equation~\ref{eq:cfr-value}. Then
  $\forall{}I\in\mathcal{I}_p$ and $\forall{}a\in{}A(I)$,
  $v^{(\sigma^{t+1}_p,\sigma_o)}(I)_a\ge{}v^{(\sigma^t_p,\sigma_o)}(I)_a$.
\end{lemma}
\begin{proof}
  We will use some additional terminology.  Let the terminal states
  reached from $I$ by action $a\in{}A(I)$ be
  \begin{align}
    Z(I,a) & \defeq \bigcup_{h \in I} Z(ha) \label{eq:infoset-terminals}
  \end{align}
  and for any descendant state of $I$, we will call
  the ancestor $h$ in $I$
  \begin{align}
    h^{I}(j) & \defeq h \in I \mbox{ s.t. } h \sqsubseteq j \label{eq:infoset-parent-history}
  \end{align}
  Let $D(I,a)$ be the set of information sets which are descendants of $I$ given action $a\in{}A(I)$, and let  $C(I,a)$ be the set of immediate children:
  \begin{align}
    D(I,a) & \defeq \{ J \in \mathcal{I}_{p(I)} \mid \exists h \in I, j \in J \mbox{ s.t. } ha \sqsubseteq j \} \nonumber \\
    C(I,a) & \defeq D(I,a) \setminus \bigcup_{J \in C(I,a), b \in A(J)} D(J,b) \label{eq:infoset-child}
  \end{align}
  Note that by perfect recall, for $J\in{}C(I,a)$, $\exists{}h\in{}I$
  such that $ha\sqsubseteq{}j$ for all $j \in J$: if one state in $J$
  is reached from $I$ by action $a$, all states in $J$ are reached
  from $I$ by action $a$. Let the distance of an information set from
  the end of the game be
  \begin{align}
    d(I) & \defeq \left\{ \begin{array}{ll}
      \max_{a \in A(I), J \in C(I,a)}(d(J) + 1) & \mbox{if $\exists a$ s.t. $C(I,a) \ne \emptyset$} \\
      0 & \mbox{if $\forall a,~C(I,a) = \emptyset$}
    \end{array} \right. \label{eq:infoset-depth}
  \end{align}
  Using this new terminology, we can re-write
  \begin{align}
    v^{\boldsymbol{\sigma}}_p(I)_a & = \sum_{h \in I} \sum_{z \in Z(h)} \pi^{\boldsymbol{\sigma}}_{-p}(z) \pi^{\boldsymbol{\sigma}}_p(z \mid ha) u_p(z) & \mbox{by Eq.~\ref{eq:cfv}} \nonumber \\
    & = \sum_{z \in Z(I,a)} \pi^{\boldsymbol{\sigma}}_{-p}(z) \pi^{\boldsymbol{\sigma}}_p(z \mid h^I(z)a) u_p(z) & \mbox{by Eqs.~\ref{eq:infoset-terminals},~\ref{eq:infoset-parent-history}} \label{eq:cfv-rewrite}
  \end{align}
  We will now show that $\forall{}i\ge{}0$
  \begin{align}
    & \forall I \in \mathcal{I}_p \mbox{ s.t. } d(I) \le i,~ \forall a \in A(I),~ v^{(\sigma^{t+1},\sigma_o)}(I)_a \ge v^{(\sigma^t_p,\sigma_o)}(I)_a \label{eq:cfr-infoset-improvement-hypothesis}
  \end{align}
  For the base case $i=0$, consider any $I\in\mathcal{I}_p$ such that
  $d(I)=0$. Given these assumptions,
  \begin{align}
    & \forall a \in A(I),~ C(I,a) = \emptyset & \mbox{by Eqs.~\ref{eq:infoset-child},~\ref{eq:infoset-depth}} \nonumber \\
    & \forall \boldsymbol{\sigma},~ \forall a \in A(I),~ \forall z \in Z(I,a),~ \pi^{\boldsymbol{\sigma}}_p(z \mid h^I(z)a) = 1 & \mbox{by Eq.~\ref{eq:pi-p-h-j}} \label{eq:cfr-infoset-improvement-no-child}
  \end{align}
  Now consider
  $v^{(\sigma^{t+1}_p,\sigma_o)}_p(I)_a$
  \begin{align}
    & = \sum_{z \in Z(I,a)} \pi^{(\sigma^{t+1}_p,\sigma_o)}_{-p}(z) \pi^{(\sigma^{t+1}_p,\sigma_o)}_p(z \mid h^I(z)a) u_p(z) & \mbox{by Eq.~\ref{eq:cfv-rewrite}} \nonumber \\
    & = \sum_{z \in Z(I,a)} \pi^{(\sigma^t_p,\sigma_o)}_{-p}(z) \pi^{(\sigma^{t+1}_p,\sigma_o)}_p(z \mid h^I(z)a) u_p(z) & \mbox{by Eq.~\ref{eq:pi-p-h}} \nonumber \\
    & = \sum_{z \in Z(I,a)} \pi^{(\sigma^t_p,\sigma_o)}_{-p}(z) \pi^{(\sigma^{t}_p,\sigma_o)}_p(z \mid h^I(z)a) u_p(z) & \mbox{by Eq.~\ref{eq:cfr-infoset-improvement-no-child}} \nonumber \\
    & = v^{(\sigma^t_p,\sigma_o)}(I)_a & \mbox{by Eq.~\ref{eq:cfv-rewrite}} \nonumber
  \end{align}
  Assume the inductive hypothesis,
  Equation~\ref{eq:cfr-infoset-improvement-hypothesis}, holds for some
  $i\ge{}0$. If $\forall{}I\in\mathcal{I}_p$, $d(I)\le{i}$,
  Equation~\ref{eq:cfr-infoset-improvement-hypothesis} trivially holds
  for $i+1$. Otherwise, consider any $I\in\mathcal{I}_p$ such that
  $d(I)=i+1$. Let $T(I,a)$ be the (possibly empty) set of terminal
  histories in $Z(I,a)$ that do not pass through another information
  set in $\mathcal{I}_{p(I)}$.
  \begin{align}
    T(I,a) & \defeq Z(I,a) \setminus \bigcup_{J \in C(I,a), b \in A(J)} Z(J,b) \label{eq:cfr-infoset-improvement-childless}
  \end{align}
  Because we require players to have perfect recall, terminal
  histories which pass through different child information sets are
  disjoint sets.
  \begin{align}
    Z(J,b) \cap Z(J',b') = \emptyset \iff J = J', b = b' \nonumber
  \end{align}
  Therefore, we can construct a partition $\mathcal{P}$ of $Z(I,a)$
  from these disjoint sets and the terminal histories $T(I,a)$ which do
  not pass through any child information set.
  \begin{align}
    \mathcal{P} \defeq & \{Z(J,b) \mid J \in C(I,a), b \in A(J)\} \cup \{T(I,a)\} & \label{eq:cfr-infoset-improvement-partition}
  \end{align}
  Note that by the induction assumption, because $d(I)=i+1$
  \begin{align}
    & \forall J \in C(I,a),~ d(J) \le i & \mbox{by Eqs.~\ref{eq:infoset-child},~\ref{eq:infoset-depth}} \nonumber \\
    & \forall J \in C(I,a), b \in A(J),~ v^{(\sigma^{t+1},\sigma_o)}(J)_b \ge v^{(\sigma^t_p,\sigma_o)}(J)_b \label{eq:cfr-infoset-improvement-decrease-depth}
  \end{align}
  Given this, we have $v^{(\sigma^{t+1}_p,\sigma_o)}(I)_a$
  \begin{align}
    = & \sum_{z \in Z(I,a)} \pi^{(\sigma^{t+1}_p,\sigma_o)}_{-p}(z) \pi^{(\sigma^{t+1}_p,\sigma_o)}_p(z \mid h^I(z)a) u_p(z) & \mbox{by Eq.~\ref{eq:cfv-rewrite}} \nonumber \\
    = & \sum_{z \in Z(I,a)} \pi^{(\sigma^t_p,\sigma_o)}_{-p}(z) \pi^{(\sigma^{t+1}_p,\sigma_o)}_p(z \mid h^I(z)a) u_p(z) & \mbox{by Eq.~\ref{eq:pi-p-h}} \nonumber \\
    = & \sum_{J \in C(I,a)} \sum_{b \in A(J)} \sum_{z \in Z(J,b)} \pi^{(\sigma^t_p,\sigma_o)}_{-p}(z) \pi^{(\sigma^{t+1}_p,\sigma_o)}_p(z \mid h^I(z)a) u_p(z) \nonumber \\
    & + \sum_{z \in T(I,a)} \pi^{(\sigma^t_p,\sigma_o)}_{-p}(z) \pi^{(\sigma^{t+1}_p,\sigma_o)}_p(z \mid h^I(z)a) u_p(z) & \mbox{by Eq.~\ref{eq:cfr-infoset-improvement-partition}} \nonumber \\
    = & \sum_{J \in C(I,a)} \sum_{b \in A(J)} \sum_{z \in Z(J,b)} \pi^{(\sigma^t_p,\sigma_o)}_{-p}(z) \pi^{(\sigma^{t+1}_p,\sigma_o)}_p(z \mid h^I(z)a) u_p(z) \nonumber \\
    & + \sum_{z \in T(I,a)} \pi^{(\sigma^t_p,\sigma_o)}_{-p}(z) \pi^{(\sigma^t_p,\sigma_o)}_p(z \mid h^I(z)a) u_p(z) & \mbox{by Eqs.~\ref{eq:pi-p-h-j},~\ref{eq:cfr-infoset-improvement-childless}} \label{eq:cfr-infoset-improvement-partition-sum}
  \end{align}
  Looking at the terms inside $\sum_J$ we have
  \begin{align}
    & \sum_{b \in A(J)} \sum_{z \in Z(J,b)} \pi^{(\sigma^t_p,\sigma_o)}_{-p}(z) \pi^{(\sigma^{t+1}_p,\sigma_o)}_p(z \mid h^I(z)a) u_p(z) \nonumber \\
    & = \sum_{b \in A(J)} \sum_{z \in Z(J,b)} \pi^{(\sigma^t_p,\sigma_o)}_{-p}(z) \sigma^{t+1}_{p}(J)_b \pi^{(\sigma^{t+1}_p,\sigma_o)}_p(z \mid h^J(z)b) u_p(z) & \mbox{by Eqs.~\ref{eq:pi-p-h-j},~\ref{eq:infoset-child}} \nonumber \\
    & = \sum_{b \in A(J)} \sigma^{t+1}_{p}(J)_b v^{(\sigma^{t+1}_p,\sigma_o)}(J)_b & \mbox{by Eq.~\ref{eq:cfv-rewrite}} \nonumber \\
    & = \boldsymbol{\sigma}^{t+1}_{p}(J) \cdot \boldsymbol{v}^{(\sigma^{t+1}_p,\sigma_o)}(J) \nonumber \\
    & \ge \boldsymbol{\sigma}^{t+1}_{p}(J) \cdot \boldsymbol{v}^{(\sigma^t_p,\sigma_o)}(J) & \mbox{by Eq.~\ref{eq:cfr-infoset-improvement-decrease-depth}} \nonumber \\
    & \ge \boldsymbol{\sigma}^t_{p}(J) \cdot \boldsymbol{v}^{(\sigma^t_p,\sigma_o)}(J) & \mbox{by Theorem~\ref{thm:rm-improvement}} \nonumber \\
    & = \sum_{b \in A(J)} \sum_{z \in Z(J,b)} \pi^{(\sigma^t_p,\sigma_o)}_{-p}(z) \pi^{(\sigma^t_p,\sigma_o)}_p(z \mid h^I(z)a) u_p(z) \nonumber
  \end{align}
  Substituting the terms back into Equation~\ref{eq:cfr-infoset-improvement-partition-sum}, we have $v^{(\sigma^{t+1}_p,\sigma_o)}(I)_a$
  \begin{align}
    \ge & \sum_{J \in C(I,a)} \sum_{b \in A(J)} \sum_{z \in Z(J,b)} \pi^{(\sigma^t_p,\sigma_o)}_{-p}(z) \pi^{(\sigma^t_p,\sigma_o)}_p(z \mid h^I(z)a) u_p(z) \nonumber \\
    & + \sum_{z \in T(I,a)} \pi^{(\sigma^t_p,\sigma_o)}_{-p}(z) \pi^{(\sigma^t_p,\sigma_o)}_p(z \mid h^I(z)a) u_p(z) & \nonumber \\
    = & \sum_{z \in Z(I,a)} \pi^{(\sigma^t_p,\sigma_o)}_{-p}(z) \pi^{(\sigma^t_p,\sigma_o)}_p(z \mid h^I(z)a) u_p(z) & \mbox{by Eq.~\ref{eq:cfr-infoset-improvement-partition}} \nonumber \\
    = & v^{(\sigma^t_p,\sigma_o)}(I)_a & \mbox{by Eq.~\ref{eq:cfv-rewrite}} \nonumber
  \end{align}
  Therefore Equation~\ref{eq:cfr-infoset-improvement-hypothesis} holds
  for $i+1$, and by induction holds for all $i$. In particular, it
  holds for $i=\max_{I\in\mathcal{I}_p}d(I)$, and applies to all
  $I\in\mathcal{I}_p$.
\end{proof}

\begin{theorem}
  \label{thm:cfr-improvement}
  Let $p$ be the player that is about to be updated in CFR or CFR$^+$
  at some time $t$. Let $\sigma^t_p$ be the current strategy for $p$,
  and $\sigma_o$ be the opponent strategy $\sigma^t_{-p}$ or
  $\sigma^{t+1}_{-p}$ used by the values defined in
  Equation~\ref{eq:cfr-value}. Then
  $u^{(\sigma^{t+1}_p,\sigma_o)}_p\ge{}u^{(\sigma^t_p,\sigma_o)}_p$.
\end{theorem}
\begin{proof}
  This immediately follows from
  Lemma~\ref{thm:cfr-infoset-improvement} and
  $u^{\boldsymbol{\sigma}}_p=v^{\boldsymbol{\sigma}}(I^\emptyset_p)_{a_0}$.
\end{proof}

As a corollary of Theorems~\ref{thm:folk-theorem}
and~\ref{thm:cfr-improvement}, when using alternating updates with
either CFR or CFR$^+$, the average strategy
$(\bar{\sigma}^{[1,t]}_1,\bar{\sigma}^{[0,t-1]}_2)$ has
$\mathcal{O}(\sqrt{t})$ exploitability. From the original papers, both
algorithms have an $\mathcal{O}(\sqrt{t})$ regret bound, and the
trailing sum in Theorem~\ref{thm:folk-theorem} is non-negative by
Theorem~\ref{thm:cfr-improvement}.  However, this only applies to a
uniform average, so we need yet another theorem to bound the
exploitability of the CFR$^+$ weighted average.

\begin{theorem}
  \label{thm:cfr-plus}
  Let $\boldsymbol{\sigma}^t$ be the CFR$^+$ strategy profile at some
  time $t$, using alternating updates so that Player 1 regret-like
  values are updated using $\boldsymbol{v}^{(\sigma^t_1,\sigma^t_2)}$
  and Player 2 regrets are updated using
  $\boldsymbol{v}^{(\sigma^{t+1}_1,\sigma^t_2)}$. Let
  $l=\max_{y,z\in{}Z}(u_1(y)-u_2(z))$ be the bound on terminal
  utilities. Then the exploitability of the weighted average strategy
  $(\frac{2}{t^2+t}\sum_{i=1}^{t}i\sigma^i_1,\frac{2}{t^2+t}\sum_{i=0}^{t-1}(i+1)\sigma^i_2)$
  is bounded by $2|\mathcal{I}|l\sqrt{k/t}$,
  where $k\defeq\max_{I}|A(I)|$.
\end{theorem}
\begin{proof}
  Consider two expanded sequences $S^1$ and $S^2$ of strategy profiles
  where the original strategy profile $\boldsymbol{\sigma}^t$ occurs
  $t+1$ times
  \begin{align}
    S^1 & \defeq \begin{array}{cccc}
      \underbrace{(\sigma^0_1, \sigma^0_2),} & \underbrace{(\sigma^1_1, \sigma^1_2), (\sigma^1_1, \sigma^1_2),} & ..., & \underbrace{(\sigma^{t-1}_1, \sigma^{t-1}_2), ..., (\sigma^{t-1}_1, \sigma^{t-1}_2)} \\
      \mbox{1 copy} & \mbox{2 copies} & & \mbox{t copies}
      \end{array} \nonumber \\
    S^2 & \defeq \begin{array}{cccc}
      \underbrace{(\sigma^1_1, \sigma^0_2),} & \underbrace{(\sigma^2_1, \sigma^1_2), (\sigma^2_1, \sigma^1_2),} & ..., & \underbrace{(\sigma^t_1, \sigma^{t-1}_2), ..., (\sigma^t_1, \sigma^{t-1}_2)} \\
      \mbox{1 copy} & \mbox{2 copies} & & \mbox{t copies}
      \end{array} \nonumber
  \end{align}
  Then with respect to $S^p$, the total Player $p$ regret for any
  information set $I$ and action $a$ is
  \begin{align}
    r^{\frac{t^2+t}{2}}_p(I)_a & \le tl\sqrt{kt} & \mbox{by CFR$^+$
      Lemma 4~\cite{15ijcai-cfrplus}} \nonumber
  \end{align}
  and the average Player $p$ regret is
  \begin{align}
    r^{\frac{t^2+t}{2}}_p & \le \frac{2}{t^2+t} \sum_{I \in \mathcal{I}_p} \max_a r^{\frac{t^2+t}{2}}(I)_a & \mbox{by CFR Theorem 3~\cite{07nips-cfr}} \nonumber \\
    & \le \frac{2}{t^2+t} \sum_{I \in \mathcal{I}_p}tl\sqrt{kt} \nonumber \\
  & \le 2|\mathcal{I}_p|l\sqrt{k/t} \label{eq:cfr-plus-regret-bound}
  \end{align}
  Because we have two sequences of profiles, we can not directly use
  Theorem~\ref{thm:folk-theorem}, but we can follow the same form as
  that proof to get
  \begin{align}
    & r^{\frac{t^2+t}{2}}_1 + r^{\frac{t^2+t}{2}}_2 \nonumber \\
    & = \max_{\sigma^*_1}\sum_{i=0}^{\frac{t^2+t}{2}-1} \biggl( u^{(\sigma^*_1,S^1_{i,2})}_1 - u^{\boldsymbol{S}^1_i}_1 \biggr) \frac{2}{t^2+t} + \max_{\sigma^*_2}\sum_{i=0}^{\frac{t^2+t}{2}-1} \biggl( u^{(\sigma^*_2,S^2_{i,1})}_2 - u^{\boldsymbol{S}^2_i}_2 \biggr) \frac{2}{t^2+t} \nonumber \\
    & = \max_{\sigma^*_1} u^{\bigl(\sigma^*_1,\bar{S^1}^{[0,\frac{t^2+t}{2}-1]}_2\bigr)}_1 + \max_{\sigma^*_2} u^{\bigl(\bar{S^2}^{[0,\frac{t^2+t}{2}-1]}_1,\sigma^*_2\bigr)}_2 - \sum_{i=0}^{\frac{t^2+t}{2}}\biggl(u^{\boldsymbol{S}^1_i}_1 + u^{\boldsymbol{S}^2_i}_2\biggr)\frac{2}{t^t+t} \nonumber \\
    & = \max_{\sigma^*_1} u^{\bigl(\sigma^*_1,\frac{2}{t^2+t}\sum_{i=0}^{t-1}(i+1)\sigma^i_2\bigr)}_1 + \max_{\sigma^*_2} u^{\bigl(\frac{2}{t^2+t}\sum_{i=1}^{t}i\sigma^i_1,\sigma^*_2\bigr)}_2 - \sum_{i=0}^{t-1}\frac{2(i+1)}{t^t+t}\biggl(u^{(\sigma^i_1,\sigma^i_2)}_1 - u^{(\sigma^{i+1}_1,\sigma^i_2)}_1\biggr) \nonumber \\
    & = \expl{\biggl(\frac{2}{t^2+t}\sum_{i=1}^{t}i\sigma^i_1,\frac{2}{t^2+t}\sum_{i=0}^{t-1}(i+1)\sigma^i_2\biggr)} - \sum_{i=0}^{t-1}\frac{2(i+1)}{t^t+t}\biggl(u^{(\sigma^i_1,\sigma^i_2)}_1 - u^{(\sigma^{i+1}_1,\sigma^i_2)}_1\biggr) \nonumber
  \end{align}
  Given Equation~\ref{eq:cfr-plus-regret-bound}, we have
  $\expl{\bigl(\frac{2}{t^2+t}\sum_{i=1}^{t}i\sigma^i_1,\frac{2}{t^2+t}\sum_{i=0}^{t-1}(i+1)\sigma^i_2\bigr)}$
  \begin{align}
    & \le 2|\mathcal{I}|_1|l\sqrt{k/t} + 2|\mathcal{I}|_2|l\sqrt{k/t} + \frac{2}{t^t+t}\sum_{i=0}^{t-1}\biggl(i+1\biggr)\biggl(u^{(\sigma^i_1,\sigma^i_2)}_1 - u^{(\sigma^{i+1}_1,\sigma^i_2)}_1\biggr) \nonumber \\
    & \le 2|\mathcal{I}|_1|l\sqrt{k/t} + 2|\mathcal{I}|_2|l\sqrt{k/t} & \mbox{by Theorem~\ref{thm:cfr-improvement}} \nonumber \\
    & = 2|\mathcal{I}|l\sqrt{k/t} \nonumber
  \end{align}
\end{proof}

\section{Conclusions}
\label{sec:conclusion}

The original CFR$^+$ convergence proof makes unsupported use of the
folk theorem linking regret to exploitability. We re-make the link
between regret and exploitability for alternating updates, and provide
a corrected CFR$^+$ convergence proof that recovers the original
exploitability bound. The proof uses a specific property of CFR and
CFR$^+$, where for any single player update, both algorithms are
guaranteed to never generate a new strategy which is worse than the
current strategy.

With a corrected proof, we once again have a theoretical guarantee of
correctness to fall back on, and can safely use CFR$^+$ with
alternating updates, in search of its strong empirical performance
without worrying that it might be worse than CFR.

The alternating update analogue of the folk theorem also provides some
theoretical motivation for the empirically observed benefit of using
alternating updates. Exploitability is now bounded by the regret minus
the average improvement in expected values. While we proved that the
improvement is guaranteed to be non-negative for CFR and CFR$^+$, we
would generally expect non-zero improvement on average, with a
corresponding reduction in the bound on exploitability.

\bibliography{burch19a}

\begin{thebibliography}{}

\bibitem[\protect\BCAY{Bowling, Burch, Johanson,\ \BBA\ Tammelin}{Bowling
  et~al.}{2015}]{SolvingHulhe}
Bowling, M., Burch, N., Johanson, M., \BBA\ Tammelin, O. \BBOP2015\BBCP.
\newblock \BBOQ Heads-up limit hold'em poker is solved\BBCQ\
\newblock {\Bem Science}, {\Bem 347\/}(6218), 145--149.

\bibitem[\protect\BCAY{{Farina}, {Kroer},\ \BBA\ {Sandholm}}{{Farina}
  et~al.}{2019}]{Farina18OnlineConvexOptimization}
{Farina}, G., {Kroer}, C., \BBA\ {Sandholm}, T. \BBOP2019\BBCP.
\newblock \BBOQ {Online Convex Optimization for Sequential Decision Processes
  and Extensive-Form Games}\BBCQ\
\newblock In {\Bem Proceedings of the Thirty-Third {AAAI} Conference on
  Artificial Intelligence}.

\bibitem[\protect\BCAY{Hart\ \BBA\ Mas-Colell}{Hart\ \BBA\
  Mas-Colell}{2000}]{hart2000simple}
Hart, S.\BBACOMMA\  \BBA\ Mas-Colell, A. \BBOP2000\BBCP.
\newblock \BBOQ A simple adaptive procedure leading to correlated
  equilibrium\BBCQ\
\newblock {\Bem Econometrica}, {\Bem 68\/}(5), 1127--1150.

\bibitem[\protect\BCAY{Nash}{Nash}{1950}]{Nash1950}
Nash, J.~F. \BBOP1950\BBCP.
\newblock \BBOQ Equilibrium points in n-person games\BBCQ\
\newblock {\Bem Proceedings of the National Academy of Sciences}, {\Bem
  36\/}(1), 48--49.

\bibitem[\protect\BCAY{Tammelin}{Tammelin}{2014}]{Tammelin14}
Tammelin, O. \BBOP2014\BBCP.
\newblock \BBOQ Solving large imperfect information games using {CFR+}\BBCQ\
\newblock {\Bem CoRR}, {\Bem abs/1407.5042}.

\bibitem[\protect\BCAY{Tammelin, Burch, Johanson,\ \BBA\ Bowling}{Tammelin
  et~al.}{2015}]{15ijcai-cfrplus}
Tammelin, O., Burch, N., Johanson, M., \BBA\ Bowling, M. \BBOP2015\BBCP.
\newblock \BBOQ Solving heads-up limit texas hold'em\BBCQ\
\newblock In {\Bem Proceedings of the Twenty-Fourth International Joint
  Conference on Artificial Intelligence (IJCAI)}.

\bibitem[\protect\BCAY{Von~Neumann\ \BBA\ Morgenstern}{Von~Neumann\ \BBA\
  Morgenstern}{1947}]{von1947theory}
Von~Neumann, J.\BBACOMMA\  \BBA\ Morgenstern, O. \BBOP1947\BBCP.
\newblock {\Bem Theory of Games and Economic Behavior}.
\newblock Princeton University Press.

\bibitem[\protect\BCAY{Zinkevich, Johanson, Bowling,\ \BBA\ Piccione}{Zinkevich
  et~al.}{2007}]{07nips-cfr}
Zinkevich, M., Johanson, M., Bowling, M., \BBA\ Piccione, C. \BBOP2007\BBCP.
\newblock \BBOQ Regret minimization in games with incomplete information\BBCQ\
\newblock In {\Bem Advances in Neural Information Processing Systems 20
  (NIPS)}, \BPGS\ 905--912.

\end{thebibliography}
\bibliographystyle{theapa}

\end{document}